\documentclass{amsart}
\usepackage{amsmath}
\usepackage{amssymb}
\usepackage{dsfont}
\usepackage{epsfig}
\usepackage[all]{xy}
\usepackage[ps2pdf,colorlinks=true,pdfpagemode=UseNone, urlcolor=blue,linkcolor=blue,citecolor=red]{hyperref}
\usepackage{url}

\newtheoremstyle{theorem}
  {10pt}		  
  {10pt}  
  {\sl}  
  {\parindent}     
  {\bf}  
  {. }    
  { }    
  {}     
\theoremstyle{theorem}
\newtheorem{thm}{Theorem}
\newtheorem{cor}[thm]{Corollary}
\newtheorem{lem}[thm]{Lemma}

\newtheoremstyle{defi}
  {10pt}		  
  {10pt}  
  {\rm}  
  {\parindent}     
  {\bf}  
  {. }    
  { }    
  {}     
\theoremstyle{defi}
\newtheorem{defn}[thm]{Definition}
\newtheorem{rem}[thm]{Remark}
\newtheorem{prop}[thm]{Proposition}
\newtheorem{conj}[thm]{Conjecture}

\title{Energy Levels of Periodic Solutions of the Circular 2+2 Sitnikov Problem}

\author{Hugo Jim\'{e}nez P\'{e}rez}
\email{hjp@fciencias.unam.mx}
\address{Dept. of Math., Univ. Aut. Metrop. Iztapalapa, San Rafael Atlixco 186,
C.P. 09430, Iztapalapa, Mexico City}

\author{Ernesto A. Lacomba}
\email{lace@xanum.uam.mx}
\address{Dept. of Math., Univ. Aut. Metrop. Iztapalapa, San Rafael Atlixco 186,
C.P. 09430, Iztapalapa, Mexico City}


\begin{document}
\maketitle

\vspace{24pt}
{\begin{minipage}{24pc}
\footnotesize{\ \ \ \ 
We introduce a restricted four body problem in a 2+2 configuration
extending the classical circular Sitnikov problem to the circular double Sitnikov problem.
Since the secondary bodies are moving on the same perpendicular line 
where evolve the primaries, almost every solution is a collision orbit.
We extend the solutions beyond collisions with a symplectic 
regularization and study 
the set of energy surfaces that contain periodic orbits.}
\end{minipage}}
\vspace{10pt}

\section{Introduction}

One of the most important problems in celestial mechanics is the
Sitnikov problem \cite{Sit1}, because this was the first restricted
three body problem where the existence of oscillatory movements was
proved, as J. Chazy predicted in 1922 \cite{Cha1}. The Sitnikov
problem is a generalization of the Macmillan problem introduced in
\cite{Mac1} which is an integrable problem and it has been studied
by several mathematicians like Alekseev \cite{Ale1} and Moser
\cite{Mos1}, Dankowicz and Holmes \cite{Dan1}, Lacomba, Llibre and P\'erez-Chavela \cite{Lac1}, 
Garc\'{\i}a and P\'erez-Chavela \cite{Gar1}, among others. Some generalizations of this
problem include the Sitnikov problem in $\mathbb R^4$ \cite{Lac1}, 
the Sitnikov problem with three equal masses \cite{Dvo1}, 
and recently the circular 4-body Sitnikov problem in a 3+1
configuration \cite{Sou1}.
In this project we study the 4-body Sitnikov problem in a $2+2$
configuration.
In this configuration and for negative values of
relative secondaries' energy $H<0$, in every solution the infinitesimal
bodies collide. 
Therefore we consider collisions as elastic bouncing 
and we are interested in periodic solutions of this type,
after applying the regularization process to continue solutions beyond collisions.

Like other restricted problems, when the masses of infinitesimal bodies
tend to zero, the system decouples and some singular terms vanish. This is
the case we study in the present project. Instead of studying continuation of 
periodic orbits from circular to elliptic cases, we are interested in the
conditions that must satisfy the values of fixed energy in order to accept resonant
torus inside the hypersurface of constant energy. In a forthcoming work, we will study 
the transcendence conditions of the total fixed energy and its impact 
on the distribution of resonant tori.

\section{The 4-body Sitnikov problem\label{chap:FBSP}}

The Sitnikov problem is a special case of the restricted three body problem
where two massive bodies with masses $m_1=m_2=\frac{1}{2}$ 
are evolving on keplerian orbits around their center of masses, and
there is an infinitesimal body that moves on the perpendicular 
straight line which passes across the center of masses
of the massive bodies.
The massive bodies are called primaries and the infinitesimal body 
is known as the secondary.
   The Sitnikov problem consists in determining the evolution of the 
body with infinitesimal mass under the attraction of primaries
with Newtonian gravitational potential.

The 4-body Sitnikov problem in 2+2 configuration (or double Sitnikov problem for short),
is realized by the addition of one more secondary body on the perpendicular straight 
line where the first secondary evolves. In the general case, the secondaries have 
different masses $m_3=\mu$ and $m_4=\nu$ with $\mu\neq\nu$ and without loss of generality
we can assume that $\nu\leq\mu\ll\frac{1}{2}$. In this way, these bodies have no effect 
on the primaries' evolution, however with big positive masses of the secondaries
the dynamics of the system is very different to the circular classical 
Sitnikov problem. Of course, the two infinitesimal bodies interact
between them under the Newtonian gravitational force.  This is the subject of the dissertation work of the first author.

The case with positive masses $\mu>0$ and $\nu>0$ will be called the reduced problem,
while the case with null masses $\mu=\nu=0$ will be 
the restricted problem and it is called the 2+2 Sitnikov problem or 
alternatively the double Sitnikov problem. This work is related
with the study of periodic orbits of the circular problem on the energy 
constant hypersurfaces.

The potential of the reduced $2+2$ problem in the general case is
\begin{eqnarray*}
  V&=&  \frac{\mu}{\sqrt{q_3^2+1/4}} + \frac{\nu}{\sqrt{q_4^2+1/4}}  +
     \frac{\mu\nu}{q_3-q_4},
\end{eqnarray*}
the vector field is 
\begin{eqnarray*}
  \mathcal M\ddot q &=& -\frac{\partial V}{\partial q}
\end{eqnarray*}
and the Hamiltonian
function is
\begin{eqnarray}
  \underline H &=& \frac{1}{2} {\bf p}^T \mathcal M^{-1} {\bf p} - \frac{\mu}{\sqrt{q_3^2+1/4}} - \frac{\nu}{\sqrt{q_4^2+1/4}} -\frac{\mu\nu}{q_3-q_4};\label{eqn:ham02}
\end{eqnarray}
where ${\bf p}=(p_3,p_4)$ are the conjugate momenta and 
$\mathcal M=\left( 
\begin{array}{cc}
  \mu & 0\\
  0 & \nu
\end{array}
\right)$ is the matrix of masses.

In general, we can 
assign a correspondence rule to secondaries' masses in the form
$\nu=f(\mu)$ such that the $\lim_{\mu\to 0}f(\mu)=0$ and then study 
the restricted problem which will depends on $\mu$ only. In our case
we consider $\nu=c\ \mu$ with $0<c\leq1$. 

\begin{rem}
The case when $1\leq c< \infty$ is obtained by interchanging $\mu$ and $\nu$, so this 
analysis is valid for $0<c<\infty$.
\end{rem}

We obtain a new Hamiltonian function that now depends on $c$ and $\mu$ as parameters. 

\begin{eqnarray}
  H &=& \frac{1}{2} {\bf p}^T \hat{\mathcal  M}^{-1} {\bf p} - \frac{\alpha}{\sqrt{q_3^2+1/4}} - 
\frac{\beta}{\sqrt{q_4^2+1/4}} -\mu\frac{\beta}{q_3-q_4};\label{eqn:ham03}
\end{eqnarray}
where  
$\hat {\mathcal M}=\left( 
\begin{array}{cc}
  \alpha & 0\\
  0 & \beta
\end{array}
\right)$, $\alpha=\frac{1}{1+c}$ and $\beta=1-\alpha$.

At this point, we are interested in restating the problem from the symplectic point of view.
We consider the open symplectic manifold $(M,\omega)$ defined as the cotangent bundle of the 
configuration space\footnote{We say that this is the \emph{cophase space}} $\mathcal Q=(\mathbb R^2\setminus \Delta)$ 
where $\Delta=\{q_3=q_4\}$ is the set of singularities
of $H$ due to collisions.
The manifold $M$ carries the standard symplectic form\footnote{
Some authors use $\omega^\prime=\sum_i dq_i\wedge dp_i$ as the standard symplectic form. 
In the formal definition we consider the standar symplectic form as the
exterior derivative of the canonical or Liouville 1-form on $M$ defined by
$\delta=\sum_i p_idq_i$. Consequently we have $\omega=d\delta$ and $\omega^\prime=-d\delta$.
}
$\omega=\sum_{i} dp_i\wedge dq_i$. We define the Hamiltonian system associated to the double circular Sitnikov problem,
by $\mathcal H=(M,\omega,X_H)$, where 
$X_H$ is the vector field associated to the Hamiltonian function $H:M \to \mathbb R$ defined by (\ref{eqn:ham03}). 

 The {\it Hamiltonian vector field} $X_H$ in local coordinates is as follows
\begin{eqnarray*}
  \dot q_3 =  \frac{1}{\alpha}p_3,&\hspace{30pt}&  \dot p_3 = -\frac{\alpha\  q_3}{\sqrt{q_3^2+\frac{1}{4}}} - \mu\frac{\beta}{(q_3-q_4)^2},\\
  \dot q_4 =  \frac{1}{\beta}p_4,&\hspace{30pt}&
  \dot p_4 = -\frac{\beta\  q_4}{\sqrt{q_4^2+\frac{1}{4}}} + \mu\frac{\beta}{(q_3-q_4)^2}.
\end{eqnarray*}

\begin{figure}
\centering
    \epsfig{file=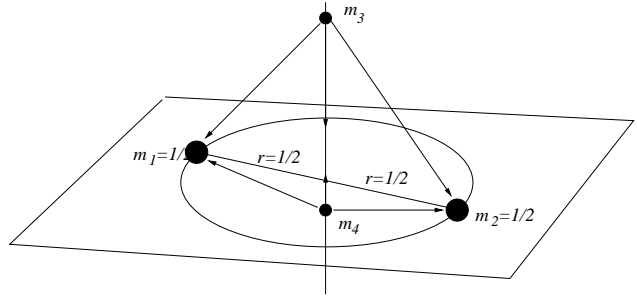,scale=0.35}
    \caption{The circular 4-body Sitnikov problem.}
    \label{fig:sit}
\end{figure}
\subsection{Regularization}
To avoid the singularity in the Hamiltonian function and in the field $X_{H}$ we 
extend analytically the equations to the hyperplane $q_3=q_4$. We perform a
symplectic regularization with the mapping $\rho:M\to M$
defined through the generating function
\begin{eqnarray*}
  W({\bf Q},{\bf p})= p_3\left( Q_4+\beta\frac{Q_1^2}{2}\right) + p_4\left( Q_4-\alpha\frac{Q_3^2}{2} \right). 
\end{eqnarray*}
Then the mapping $\rho: ({\bf Q},{\bf P})\mapsto ({\bf q},{\bf p})$ will be 
\begin{eqnarray}
  q_3 = Q_4+ \beta\frac{Q_3^2}{2},&\hspace{30pt}& p_3 = \alpha P_4 + \frac{P_3}{Q_3},\nonumber\\
  q_4 = Q_4- \alpha\frac{Q_3^2}{2},&\hspace{30pt}& p_4 =\beta P_4 - \frac{P_3}{Q_3}.\label{eqn:coord}
\end{eqnarray}

It is not difficult to show that $\rho ^*(\sum_{i}^{}dp_i\wedge dq_i)=\sum_{i}^{}dP_i\wedge dQ_i$
and, therefore $\rho \in Sp(M,\omega)$ 

Also we consider the time rescaling
\begin{eqnarray}
	\frac{dt}{d\tau} &=& \alpha\beta\ Q_3^2.\label{eqn:newtime}
\end{eqnarray}
We will obtain a new function depending
on the fixed value $h=constant$ as a parameter
in the following way: first we apply the change of coordinates 
defined by $\rho({\bf Q},{\bf P})=({\bf q},{\bf p})$ and then the Hamiltonian 
function $H({\bf q},{\bf p})=H(\rho({\bf Q},{\bf P}))$. 
Since $\rho:M\to M$ is a symplectomorphism
then the function in the new variables $\overline H({\bf Q},{\bf P})=H(\rho({\bf Q},{\bf P}))$ is again a 
Hamiltonian function.  We fix the value of the function $h=H(\rho({\bf Q},{\bf P}))$
rearrange the terms and multiply by the rescaling time to obtain $\frac{dt}{d\tau}(H\circ \rho-h)=0$.

The regularized Hamiltonian function is
\begin{eqnarray*}
  L &=& \alpha\beta\ Q_3^2\left( H-h \right)\circ \rho
\end{eqnarray*}
this Hamiltonian function depends on $\alpha,h$ as parameters and is valid only in the 
energy level $L= 0$ for each $h$ fixed.
Specifically, if ${\bf z}=(Q_3, Q_4,P_3,P_4)$
we apply the mapping $\rho$ and the Hamiltonian function 
\begin{eqnarray*}
   \left( H-h \right)\circ \rho({\bf z}) &=& 
      \frac{1}{2}\left( P_4^2 +\frac{1}{\alpha\beta}\frac{P_3^2}{Q_3^2} \right)
           - \frac{2\alpha}{\sqrt{\left(2 Q_4+\beta Q_3^2\right)^2+1}}\\
   & &  - \frac{2\beta}{\sqrt{\left(2 Q_4-\alpha Q_3^2\right)^2+1}} -     
        \mu\frac{2\beta}{Q_3^2}-h,
\end{eqnarray*}
 and after applying the time rescaling  (\ref{eqn:newtime}) we have 
\begin{eqnarray}
  L &=& \frac{1}{2} \left( \alpha\beta \ P_4^2 Q_3^2 + P_3^2\right) - 2\alpha\beta^2 \mu \label{eqn:ham:red} \\
  & & - \alpha\beta\ Q_3^2\left[ \frac{2\alpha}{\sqrt{\left(2 Q_4+\beta Q_3^2\right)^2+1}}
  + \frac{2\beta}{\sqrt{\left(2 Q_4-\alpha Q_3^2\right)^2+1}} +h \right].\nonumber
\end{eqnarray}
We write $L_h({\bf z},\mu)=L({\bf z},\mu;h)$, and the dependence on the parameter $\alpha$ 
will be dropped because in this paper we only consider 
 $c=1$ as we will see below.

We call to the triplet $\mathcal L_h=(M, \omega, X_{L_h({\bf z},\mu)})$ the regularized
system, where 
 $ X_{ L_h}$ is the regularized Hamiltonian
field. 

Although the form of the new Hamiltonian function is quite complicated, the advantage is
that this function and the Hamiltonian vector field $X_{L_h} $ are
regular on the boundary of $M$ (specifically on the set $\Delta$). Now we can obtain the limit when the mass $\mu$ goes 
to zero
\begin{eqnarray*}
  \lim_{\mu\to 0} L_h({\bf z},\mu) &=& L_h({\bf z},0),
\end{eqnarray*}
and the effect is that the term $-2\alpha\beta^2\mu$ vanish. We can reverse the process, 
and since $\alpha\beta Q_3^2$ is not identically zero, we recover the Hamiltonian function
in the original coordinates  as follows
\begin{eqnarray}
  H &=& \frac{1}{2} {\bf p}^T \hat {\mathcal M}^{-1} {\bf p} - \frac{\alpha}{\sqrt{q_3^2+1/4}} -
 \frac{\beta}{\sqrt{q_4^2+1/4}}.\label{eqn:ham01}
\end{eqnarray}

Rewriting $\underline H=H(1+c)$ and considering the momenta $p_i=\dot q_i$ for $i=3,4$,
the original Hamiltonian function for the restricted case is 
\begin{eqnarray}
  \underline H 
  &=& \left( \frac{1}{2} p_3^2  - \frac{1}{\sqrt{q_3^2+1/4}}\right) + c\left( \frac{1}{2}p_4^2  - \frac{1}{\sqrt{q_4^2+1/4}}\right).\label{eqn:ham04}
\end{eqnarray}

 As we can see,
the Hamiltonian function  (\ref{eqn:ham04}) corresponds to two uncoupled Sitnikov problems.
Figure \ref{fig:sit} shows a diagram of it.

The regularization permit us to continue analytically the solutions 
to the collision manifold $q_3=q_4$, however, if we want to study the problem 
as two uncoupled Sitnikov problems, we must give additional hypothesis in 
order to glue the solutions (contained in the same energy level) in a smooth 
way beyond the collision.
When the secondaries have different positive masses $0<\mu$, $0<\nu=c\mu$ with 
$0<c<1$,
the elastic bouncing condition
\begin{eqnarray}
	\hat v_3 - \hat v_4 = -(v_3 - v_4), \label{eqn:vel:rel}
\end{eqnarray}
and the conservation of linear momentum 
\begin{eqnarray}
	\hat p_3 + \hat p_4 = p_3 + p_4,\label{eqn:mom:tot}
\end{eqnarray}
implies
 the interchange of conjugate momenta at collision.
(Here
the terms
 $\hat p_i$ and $\hat v_i= \hat {\dot q}_i$, $i=3,4$ are respectively 
the momenta and velocities   of the bodies after collision.)
 Therefore, for certain values of the masses $\mu$ and $\nu$ there
  will be a discontinuity in the solutions of the original system
 $(M,\omega,X_H)$ for the restricted case.

Using the equations (\ref{eqn:vel:rel}) and (\ref{eqn:mom:tot}) we can see the 
behavior of the system with Hamiltonian function (\ref{eqn:ham04}) beyond collision.
Writing expression (\ref{eqn:mom:tot}) in the tangent space $T_{\bf q}\mathcal Q$ we have 
\begin{eqnarray}
	\alpha \hat v_3 + \beta \hat v_4 = \alpha v_3 + \beta v_4,\label{eqn:mom:tot2}
\end{eqnarray}
 
Let us solve for $\hat v_4$ in (\ref{eqn:vel:rel}) and substitute in  (\ref{eqn:mom:tot2}).
Now, let us solve for $\hat v_3$ in the resulting equation to get 
\begin{eqnarray}
	 \hat v_3 = v_3 -2\beta(v_3-v_4),
\end{eqnarray}
In the same way, we obtain 
\begin{eqnarray}
  \hat v_4 = v_4 + 2\alpha(v_3-v_4). \label{eqn:v3:post}
\end{eqnarray}

We write the system of two equations in matrix notation as $\hat{\bf v}=A{\bf v}$ where ${\bf v}=(v_3,v_4)^T$
and 
\begin{eqnarray}
A=
\left(
\begin{array}{cc}
   1-2\beta & 2\beta\\
   2\alpha  &  1-2\alpha
\end{array}
\right).
\end{eqnarray} 
In the cotangent space $T_{\bf q}^*\mathcal Q$ this system is written as
$\hat {\bf p}=\mathcal MA\mathcal M^{-1}{\bf p}$. Making the computations 
we obtain $\mathcal MA\mathcal M^{-1}=A^T$.
Then the condition to continue the solutions in a 
smooth way beyond collisions using the transition conditions
$\hat p_3=p_4$, $\hat p_4=p_3$, $\hat v_3 = v_4$, and
$\hat v_4 = v_3$,
is 
\begin{eqnarray}
A=A^T=
\left(
\begin{array}{cc}
   0 & 1\\
   1  &  0
\end{array}
\right). 
\end{eqnarray}
Thus, this is possible if and only if $\alpha=\beta=1/2$ or equivalently $\mu=\nu$.
We have proved the following
\begin{prop}
In the circular double Sitnikov problem if $\mu=\nu$ then the flow $\varphi_t(x)$ of the limiting
case $\mu \to 0$ can be extended to a complete flow in a natural way considering crossing beyond 
collisions instead of elastic bouncing by means of the identification
\begin{eqnarray}
   q_3-q_4 \mapsto -(q_3-q_4)
\end{eqnarray}
which extends the Hamiltonian system to the whole phase space.
\end{prop}

\begin{figure}
\centering
 \epsfig{file=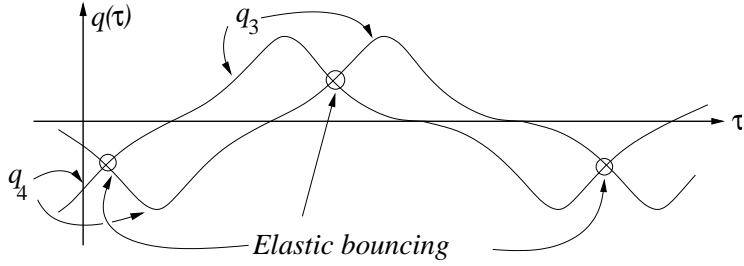,scale=0.55}
\caption{Solutions beyond collisions}
\end{figure}

On the other hand if we start with the conditions $\hat p_4 = p_3$ and $\hat p_3=p_4$,
and expanding $\hat{\bf p}=A^T{\bf p}$ we arrive to
\begin{eqnarray*}
	p_4 &=& (1-2\beta)p_3 + 2\alpha p_4,\\
	p_3 &=& 2\beta p_3 + (1-2 \alpha) p_4,
\end{eqnarray*}
from where
\begin{eqnarray*}
	\left(1-2\beta\right)p_3 + \left(2\alpha-1 \right)p_4 =0,
\end{eqnarray*}
and finally, since $1-2\beta=2\alpha-1$ we get
\begin{eqnarray}
	(2\alpha-1)(p_3+p_4)=0. \label{eqn:cond:pi}
\end{eqnarray}
This expression implies that if $\alpha\neq \frac{1}{2}$ (i.e. if $\mu\neq\nu$) 
then the solution can be continued by changing signs 
$\hat p_3=-p_3$ and $\hat p_4=-p_4$
only if $p_3+p_4=0$. 
This situation is equivalent to reversing the solution after collision
and this is the classical conception of elastic bouncing.
In this case we will be 
interested in solutions that cross the $2$-dimensional plane 
\begin{eqnarray*}
 \mathcal P=\{q_3=q_4\}\cap\{  p_3=-p_4\},
\end{eqnarray*}
and this will be studied in a forthcoming paper.

For convenience, we will analyze the problem in the original equations 
and this implies to use the expression (\ref{eqn:ham04}) with $c=1$.

\section{Action-angle coordinates and analytical solutions.}
Following Hoffer and Zehnder \cite{Hof1}, we observe that if $(M,\omega)$ is an exact symplectic manifold of dimension $2n$, there exists 
a $1$-form such that $\omega=d\lambda$. For every symplectomorphism 
$\varphi\in Sp(M)$, the
$1$-form $\lambda-\varphi^*\lambda$ is closed and by the Poincar\'e lemma, 
locally there exists $f:M\to\mathbb R$ such that $df=\lambda-\varphi^*\lambda$.

Integrating over a simple closed curve $\gamma$ we have $\int_\gamma \lambda
=\int_\gamma \varphi^*\lambda$, and finally we obtain
\begin{eqnarray*}
	\int_\gamma \lambda=\int_{\varphi(\gamma)} \lambda.
\end{eqnarray*}
We define the action on a simple closed curve $\gamma$ as 
$J(\gamma)=\int_\gamma \lambda$. An 
immediate consequence of the above computations is that the action is 
invariant under symplectomorphisms.

This invariant property on simple closed curves permit us to construct a symplectomorphism 
for every periodic integrable Hamiltonian system $(M,\omega,H)$
that only depends on the 
values of the momentum map (although the term ``action-angle'' coordinates is
generalized to non periodic Hamiltonian systems on non compact symplectic manifolds). 

Let $H=h=h_3+h_4$ be the separable Hamiltonian function associated to the
circular double Sitnikov problem. We can consider $h_3<0$ and $h_4<0$
and two closed orbits $\gamma_{h_3},\gamma_{h_4}$ associated to the relative
energies $h_3$ and $h_4$ respectively. Then, there exist a symplectic 
change of coordinates $\phi:M\to M$ 
where the transformed Hamiltonian function only depends on the action of the
(simple closed) integral curves of each fundamental field. These coordinates are
called {\it action-angle} coordinates, and are defined by 
\begin{figure}[ht]
  \centering
    \includegraphics[scale=.45]{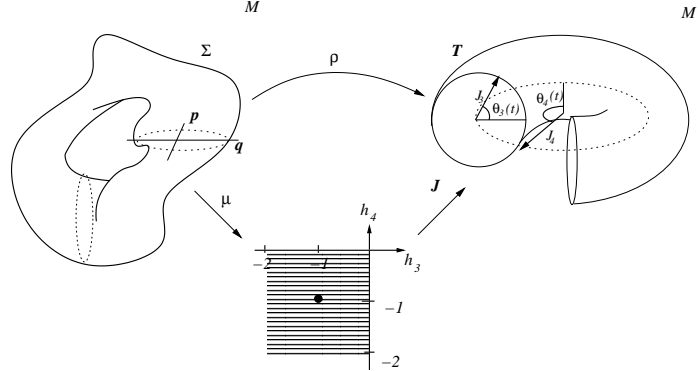}
    \caption{Action-angle coordinates.}
    \label{fig:act}
\end{figure}

\begin{eqnarray*}
	J( \gamma_{h_i})= \frac{1}{2\pi}\int_{\gamma_{h_i}} \lambda,\hspace{30pt}
	\theta(h_i) = \frac{1}{\Omega(h_i)}t + \theta_0,
\end{eqnarray*}
where $\Omega(h_i)=\frac{\partial J}{\partial h_i}$. We have the following.

\begin{thm}
  The action-angle coordinates for the double Sitnikov problem takes the form
  \begin{eqnarray}
    J(h_i) &=& \frac{\sqrt{2}}{\pi}\left( 2 E(k_i) - K(k_i) - \Pi(2k_i^2,k_i) \right), \label{eqn:Jh1}\\
    \theta_i(t;h_i) &=& \frac{1}{\Omega_i} t(\nu,k) + \theta_{0,i},
  \end{eqnarray}
  where $\Omega_i=\frac{\sqrt{2}}{4\pi(1-2 k_i^2)}(2E(k_i)-K(k_i)+\Pi(2k_i^2,k_i))$
  is the return time of the secondaries, $k_i=\frac{\sqrt{2+h_i}}{2}$ and
   $\theta_{0,i},i=3,4,$ are constants determined by the initial conditions.
  \label{teo:action}
\end{thm}
\begin{proof}
The action is defined as the integral $J_i (h_i)=\frac{1}{2\pi} \oint p_i\ dq_i$
on a complete period. Since each $h_i$ is symmetric in $p_i$ and $q_i$, we can integrate over a quarter of period and multiply
it by four. 
\begin{eqnarray}
  J_i(h_i)&=& \frac{2\sqrt{2}}{\pi}\oint_0^{q_{max}}{ \sqrt{ h_i+\frac{1}{\sqrt{q_i^2+\frac{1}{4}}}}}\ dq_i,
	\label{eqn:int1}
\end{eqnarray}
where $q_{max}$ is obtained when $p_i=0$, then 
$q_{max}= \sqrt{\frac{1}{(-h_i)^2}-\frac{1}{4}}.$

We construct a suitable change of variables $(q,p)\mapsto(y,s)$ that 
``normalizes'' the integral, such that the following conditions hold: 

a) the Hamiltonian function takes the form $\hat h=y^2/2 + a\ s^2$, 

b) for $q=0$ we require that $s=0$,

c) for $q=q_{max}$ we require that $s=1$ and $p=0$.

  The suitable change of variables 
  \begin{eqnarray}
    -\frac{1}{\sqrt{q_i^2+1/4}} &=& (2+h_i)s_i^2-2 ,\label{eqn:chan1} 
  \end{eqnarray}
  transforms the integrand of (\ref{eqn:int1}) in $\sqrt{h_i+\frac{1}{\sqrt{q_i^2+1/4}}}=\sqrt{2+h_i}\sqrt{1-s^2}$. We write $k_i=\sqrt{2+h_i}/2$ and solve (\ref{eqn:chan1}) for $q_i$,
  then we compute $dq_i$ to obtain
\begin{eqnarray}
\sqrt{h_i+\frac{1}{\sqrt{q_i^2+1/4}}}dq_i=2 k_i^2 \frac{\sqrt{1-s^2} }{\sqrt{1-k_i^2 s^2}(1-2 k_i^2 s^2)^2}ds.
\end{eqnarray}
The integral (\ref{eqn:int1}) takes the form 
\begin{eqnarray}
	J_i(h_i)&=& \frac{4\sqrt{2}}{\pi} k_i^2 \oint_0^1 \frac{\sqrt{1-s^2}ds}{\sqrt{1-k_i^2s^2}\left( 1-2k_i^2s^2 \right)^2}.\label{eqn:Jh}
\end{eqnarray}
This is a general complete elliptic integral. It is possible to write any elliptic integral
in terms of algebraic rational functions of the independent variable, and the 
elliptic integrals of first, second and third
kinds \cite{Han1}. First we integrate by parts

\begin{eqnarray*}
	\int \frac{\sqrt{1-s^2}ds}{\sqrt{1-k_i^2s^2}\left( 1-2k_i^2s^2 \right)^2}&=& 
	 \frac{s\sqrt{1-s^2}\sqrt{1- k^2 s^2}}{1 -2 k^2 s^2}  
	  +\int \frac{s^2\sqrt{1-k^2 s^2}\ ds}{\sqrt{1-s^2}(1-2k^2 s^2)}. 
\end{eqnarray*}
We rewrite the last integral in the form 
\begin{eqnarray*}
	  \frac{1}{2k^2}\int \frac{2 k^2 s^2\sqrt{1-k^2 s^2}ds}{\sqrt{1-s^2}(1-2k^2 s^2)} &=&
            -\frac{1}{2k^2}E(s,k_i) + \frac{1}{2k^2}\int\frac{\sqrt{1-k^2 s^2}ds}{\sqrt{1-s^2}(1-2k^2 s^2)},
\end{eqnarray*}
and again the last integral will be rewritten as
\begin{eqnarray*}
    \frac{1}{4k^2}\int\frac{[1+ (1-2k^2 s^2)]ds}{\sqrt{1-k^2 s^2}\sqrt{1-s^2}(1-2k^2 s^2)} &=&
              \frac{1}{4k^2}F(s_i,k_i) + \frac{1}{4k^2}\Pi(2k^2,s_i,k_i).  
\end{eqnarray*}
Putting everything together we get 
\begin{eqnarray*}
	 \int \frac{\sqrt{1-s^2}ds}{\sqrt{1-k_i^2s^2}\left( 1-2k_i^2s^2 \right)^2}ds&=& 
	 \frac{s\sqrt{1-s^2}\sqrt{1- k^2 s^2}}{1 -2 k^2 s^2} - \\
	 & & \frac{1}{4 k^2} \left( 2E(s,k_i)- \Pi(2k_i^2,s,k_i) - F(s,k_i) \right).
\end{eqnarray*}
Finally, evaluating on the integration limits, we obtain (\ref{eqn:Jh1}).

Now, the values of the period for each one of the secondaries is obtained
 in a straightforward way just by calculating the derivatives:
\begin{eqnarray}
   \frac{T(h_i)}{2\pi}=\frac{\sqrt{2}}{\pi} \frac{1}{4(1-2k^2)} \left( 2E(k_i)-K(k_i)+\Pi(2k_i^2,k_i)\right).
\end{eqnarray}
\end{proof}

Note that the solution of the angle coordinates uses the time $t=t(\nu_i,k_i)$
that is not computed yet. This variable is obtained directly from the solution 
of the classical Sitnikov problem exposed by Belbruno, Oll\'e and Llibre in 
\cite{Bel1}.


\begin{thm}
	The solutions for the circular double Sitnikov problem can be written as

\begin{eqnarray*}
	\sigma(t)= 
	\left( \frac{k_3\ s(\nu_3)\ d(\nu_3)}{1-2 k_3^2\ s^2(\nu_3)},2\sqrt{2}k_3\ c(\nu_3),
	\frac{k_4\ s(\nu_4)\ d(\nu_4)}{1-2 k_4^2\ s^2(\nu_4)},2\sqrt{2}k_4\ c(\nu_4) \right),
\end{eqnarray*}
where $\nu_i$ are functions of $t$ obtained inverting the function
\begin{eqnarray}
	t&=& \int \frac{\sqrt{2}}{4(1-2k^2sn(\nu_i)^2)^2}d\nu_i, \label{eqn:tint}
\end{eqnarray}
and $s(\nu_i)\equiv sn(\nu_i(t),k_i)$,  $c(\nu_i)\equiv cn(\nu_i(t),k_i)$,  
$d(\nu_i)\equiv dn(\nu_i(t),k_i)$ are the sine, cosine, and delta amplitude 
Jacobi elliptic functions, and $k_i=\frac{\sqrt{2+h_i}}{2}$ for $i=3,4$.\label{teo:sol}
\end{thm}

\begin{proof}
Since the solution of Newtonian Hamiltonian systems with
one degree of freedom that have the form $H=p^2/2-V(q)$ is
\begin{eqnarray}
 t+t_0=\int_{q_0}^q \frac{dq}{\sqrt{2(h+V(q))}}
\end{eqnarray}
we can use the change of variables 
(\ref{eqn:chan1}) in the differential 
\begin{eqnarray}
    dt= \frac{dq}{\sqrt{h+\frac{1}{\sqrt{q^2+1/4}}}}= \frac{1}{2\sqrt{2}} \frac{ds}{\sqrt{1-s^2}\sqrt{1-k^2s^2}(1-2k^2s^2)^2},
\end{eqnarray}
and rescale the time for $\frac{dt}{d\nu}d\nu$ to obtain 
\begin{eqnarray}
   \frac{dt}{d\nu} = \frac{1}{2\sqrt{2}(1-2k^2s^2)^2}\hspace{20pt} {\rm and} \hspace{20pt}
    d\nu = \frac{ds}{\sqrt{(1-s^2)(1-k^2s^2)}}.
\end{eqnarray}
We solve first the second integral directly because this is an elliptic
integral of the first kind. The solution is $s=sn(\nu,k)$ and this solution is 
substituted in the first integral
to obtain the relation between the old and new times
\begin{eqnarray}
   t = \int \frac{d\nu}{2\sqrt{2}(1-2k^2sn(\nu,k)^2)^2}.
\end{eqnarray}
Finally, solving for $q_i$ in (\ref{eqn:chan1}) and substituting the solution
for $s=sn(\nu,k)$ we obtain 
\begin{eqnarray}
   q_i &=& k\frac{s\sqrt{1-k^2s^2}}{1-2k^2s^2} =  k \frac{sn(\nu,k)dn(\nu,k)}{1-2k^2sn(\nu,k)^2}.
\end{eqnarray}
In order to obtain the conjugate momentum, we differentiate 
$p_i=\frac{dq_i}{dt}=\frac{d q_i}{d\nu}\frac{d\nu}{dt}$ to get
\begin{eqnarray}
p_i=\frac{d q_i}{d\nu}\frac{d\nu}{dt}=2\sqrt{2}kcn(\nu,k).
\end{eqnarray}
\end{proof}


It is possible to integrate the expression 
(\ref{eqn:tint}) with elliptic functions
and elliptic integrals to obtain  
\begin{eqnarray}
	t&=& \frac{\sqrt{2}}{8(1-2k^2)}\left[2 E(\nu) - \nu +\Pi(\nu,2k^2)
	-4k^2\frac{sn(\nu)cn(\nu)dn(\nu)}{1-2k^2sn(\nu)^2}\right] + C,\hspace{10pt}\label{eqn:timet}
\end{eqnarray}
where $C$ is an arbitrary constant of integration. In \cite{Cor1} the reader will 
find a nice and complete study of this function. 

From equation (\ref{eqn:timet}) the action-angle coordinates are completely
described as

  \begin{eqnarray*}
    J(h_i) &=& \frac{\sqrt{2}}{\pi}\left( 2 E(k_i) - K(k_i) - \Pi(2k_i^2,k_i) \right), \\
    \theta_i(\nu) &=& \frac{\pi}{2} \left(\frac{ 
   2 E(\nu) - \nu +\Pi(\nu,2k^2)
	-4k^2\frac{sn(\nu)cn(\nu)dn(\nu)}{1-2k^2sn(\nu)^2}
    }{(2E(k_i)-K(k_i)+\Pi(2k_i^2,k_i))}\right)  + \theta_{0,i}.
  \end{eqnarray*}

  From these expressions it is possible to deduce when $\nu_i=K(k_i)$
  that $\theta_i(K)=\frac{\pi}{2}+\theta_{0,i}$, therefore $\nu_i=4K(k_i)$
  implies $\theta_i(4K)=2\pi+\theta_{0,i}$. Consequently, the
  solutions of the Hamiltonian subsystems generated by $H_i=h_i$ 
  have period $\nu=4K(k_i)$ in the $\nu$ variable and 
  $T(h_i)=\frac{2\pi}{\Omega_i}$ in the $t$ variable. This fact
  will be usefull in the analysis of periodic orbits for the
  circular double Sitnikov problem. 

Just one comment before to passing to the study of the level sets and
periodic orbits of our problem. As the reader can observe, 
it is usually difficult to find the inverse 
transformation $\phi^{-1}$ to have the new Hamiltonian function
explicitly $\hat H(J,\theta)=(H\circ \phi^{-1})(J,\theta)$, however
since $\phi:({\bf q},{\bf p})\mapsto (J,\theta)$ is a symplectomorphism,
in particular $\det\left( \frac{\partial \phi}{\partial p,q} \right)=1$
and applying the inverse function  theorem, locally $\phi^{-1}$ always exists.

\section{Hyper-surfaces of fixed energy}
In this section we describe the topology of the level sets for the 
Hamiltonian function (\ref{eqn:ham04}) with $c=1$, in terms 
of the momentum map 
and its image. Since for every point $x\in {\rm Img}(\mu)$ its fiber $\mu^{-1}(x)$ is 
a Lagrangian submanifold, we can decompose ${\rm Img}(\mu)$ in smooth 
subsets with boundary to
construct the foliations  by hypersurfaces of constant energy.

\subsection{Completely integrable Hamiltonian systems}

A completely integrable Hamiltonian system 
is a Hamiltonian system $(M,\omega,X_H)$ and a set $F=\{F_i\}_{i=1}^n$ of 
first integrals for $H=F_1$ which are functionally independent and 
they are in \emph{involution} (i.e., $\{ F_i, F_j \} = 0$ where $\{f,g\}=\omega(X_f,X_g)$ denotes
the Poisson bracket). In this case we call the set $(M,\omega,X_H,F)$ 
a completely integrable system in de sense of Liouville.

%

In this context, the circular double Sitnikov problem is a completely 
integrable Hamiltonian system. Every first integral $F_i$, $2\le i \le n$
generates an one-parameter family of symplectomorphisms by means of
the exponential map. This one-parameter family can be realized as a
Lie group $G$ acting on the manifold $M$.

%
%
%

We say that the action is symplectic if for every $g\in G$,
we have that the flow $\varphi_g$ is a symplectomorphism.
     Additionally, we say that the action is a Hamiltonian 
     action if each of the fundamental fields is a Hamiltonian vector field. 
More specifically,  the action $\varphi_t$ is Hamiltonian if there exists 
a map $\mu:M\to\mathfrak g^*$, from the symplectic manifold to the dual of 
the Lie algebra $\mathfrak g=Lie(G)$ such that for every 
$X\in \mathfrak g$, the component
$ \mu^X(p):= \langle\mu(p),X\rangle $
of $\mu$ along $X$  
and  for the fundamental vector field $X^\sharp$  on $M$ generated by the 1-parameter
subgroup $G^0=\{ exp(tX)| t\in\mathbb R\}\subseteq G$, the relation 
\begin{eqnarray*}
   d\mu^X=i_{X^\sharp}\omega
\end{eqnarray*}
holds, i.e., the function $\mu^X$ is a Hamiltonian function for $X^\sharp$
and $\mu\circ \varphi_g = {\rm Ad}_g^*\circ \mu$, for all $g\in G$.

     Each fundamental field of an integrable Hamiltonian system is generated 
     by one first integral $F_i$,  such that $\{F_i,F_j\}=0$ for $1\leq i\leq k$. 
    The application 
     \begin{eqnarray*}
       \mu &=& (H=F_1,\dots,F_k):M\to \mathfrak g^*\cong \mathbb R^k
     \end{eqnarray*}
     is called the {\it momentum map} and is defined from the symplectic
     manifold to the dual of the Lie algebra associated to the Lie group that acts
     on $M$. If $k<n$ the system is partially integrable, however if $k=n$
     the system is Liouville integrable or completely integrable.

We consider the momentum map $\mu=(H_3,H_4):(M,\omega)\to \mathbb R^2$
defined by 
\begin{eqnarray}
      \mu &=& \left( \frac{1}{2}p_3^2- \frac{1}{\sqrt{q_3^2+r(t)^2}}
      , \frac{1}{2}p_4^2- \frac{1}{\sqrt{q_4^2+r(t)^2}}\right).\label{eqn:mom}
\end{eqnarray}
\begin{figure}[ht]
  \centering
    \includegraphics[scale=0.8]{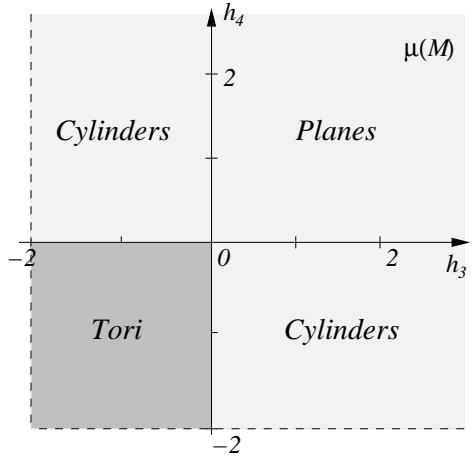}
    \caption{Values of the momentum map}
    \label{fig:fol1}
\end{figure}
     In our case $G$ has three possibilities 
\begin{itemize}
 \item[$\imath$)] $G=\mathbb S^1\times\mathbb S^1$ if $\mu({\bf q},{\bf p})$ belongs to the third quadrant,
 \item[$\imath\imath$)] $G=\mathbb S^1\times\mathbb R$ if $\mu({\bf q},{\bf p})$ belongs to the second or fourth quadrant,
 \item[$\imath\imath\imath$)] $G=\mathbb R\times\mathbb R$ if $\mu({\bf q},{\bf p})$ belongs to the first quadrant.
\end{itemize}
In all three cases $\mathfrak g^*\cong \mathbb R^2$ is obtained.

It is well-known that the inverse image $\mu^{-1}(x)$
of each $x\in {\rm Img}(\mu)$ is a Lagrangian sub-manifold of $(M,\omega)$.
If it is
a compact set, it will be isomorphic to a torus. In other cases, it would be
isomorphic to cylinders or planes
according with the region where the point $x$ lies (see Figure \ref{fig:fol1}).

 In what follows we prove a result related to the image of hypersurfaces of constant
energy of completely and separable integrable Hamiltonian systems under its momentum map.

\begin{lem}
	Let $(M,\omega)$ be an exact symplectic manifold of dimension
	$2n$, and $\mathcal H=(M,\omega,X_H)$ be a Hamiltonian system over 
	$(M,\omega)$ with Hamiltonian 
	vector field defined 
	by $i_{X_H}\omega=dH$. Suppose that there exists a symplectomorphism
	$\rho:M\to M$ such that the new Hamiltonian function $F=H\circ \rho$
	is separable. Then there exists a Lagrangian fibration $\pi:M\to\mathbb R^n$
	such that the hypersurfaces of constant energy $H$ map to hyperplanes 
	which are perpendicular to the vector $\mathds 1=(1,1,\dots,1)
	\in\mathbb R^n$.
\end{lem}
\begin{proof} 
Since there exists $\rho:M\to M$, symplectomorphism such that 
$F=H\circ \rho$ is separable then there exist global coordinates $(Q,P)$ 
where $F(Q,P)=H(\rho(Q,P))$ separates in the form
\begin{eqnarray*}
	F(Q,P)=F_1(Q_1,P_1)+ \dots+F_n(Q_n,P_n),
\end{eqnarray*}
and $F_i(Q_i,P_i)=constant$ are $n$ first integrals for $X_F$. Moreover
$\{F,F_i\}=0$ for $i=1,\dots,n$. The Hamiltonian system is integrable 
by quadratures and we can consider 
the combined flow $\varphi^{\bf t}$ of all the Hamiltonian vector fields 
$X_{F_i}$ as a Hamiltonian action of the Lie group 
$G=\mathbb R^k\times\mathbb T^{n-k}$
on $(M,\omega)$ for some $1\leq k\leq n$.

The Hamiltonian action $G\times M\to M$ induces a momentum map 
\begin{eqnarray*}
	\mu = (F_1,F_2,\dots,F_n):M\to \mathfrak g^*,
\end{eqnarray*}
where $\mathfrak g^*\cong \mathbb R^n$ is the dual of the Lie algebra
associated to $G=\varphi^{\bf t}$. Its image ${\rm Img}(\mu)\subset\mathbb R^n$ 
is a convex 
polyhedron or cone whose vertices are the extremal values of $\mu$ as was studied 
by Guillemin and Stenberg in \cite{Gui1}.
The image $\mu(\Sigma)$ of every regular hypersurface of constant 
energy $\Sigma_h=F^{-1}(h)$ under the momentum map is a convex subset $\mathcal K_h$
of the linear affine subspace of codimension 1 of $\mathfrak g^*$
\begin{eqnarray*}
	x_1+x_2+ \dots+x_n -h&=& 0,
\end{eqnarray*}
where ${\bf x}=(x_1,\dots,x_n)\in\mathbb R^n$. We can write $\langle{\bf x},
\mathds 1 \rangle -h=0$, and in particular we have
\begin{eqnarray*}
	\mathcal K_h:=\{ {\bf x}\in\mathbb R^n|
	{\bf x}\in ({\rm Img}(\mu)\cap\{\langle {\bf x},\mathds 1\rangle
	=h\}) \}\subset 
	\{ {\bf x}\in\mathbb R^n| \langle{\bf x},\mathds 1 \rangle =h\}.
\end{eqnarray*}
Then $\pi:=\mu\circ\rho^{-1}:M\to \mathbb R^n$ is the smooth map we are 
looking for. 

\begin{eqnarray*}
\xymatrix{
  M \ar[r]^{\rho} \ar[rd]^{\pi}& M \ar[d]^{\mu}\\
   &  \mathfrak g^*    }
\end{eqnarray*}

Finally, we know that the fibers $\mathfrak L=\mu^{-1}(x)$ with $x\in\mathfrak g^*$
are Lagrangian submanifolds of $M$ for every $x\in {\rm Img}(\mu)$, this implies that
$\omega|_{\mathfrak L}\equiv 0$. Since $\rho^{-1}\in Sp(M,\omega)$ then 
\begin{eqnarray}
    \omega(x,y)=\omega(\rho^{-1}(x),\rho^{-1}(y))=0, \hspace{30pt} \forall\ x,y\in \mathfrak L
\end{eqnarray}
therefore $\hat{\mathfrak L}=\rho^{-1}(\mathfrak L)$ is a Lagrangian submanifold. 
We conclude that $\pi:M\to \mathbb R^n$ is also a Lagrangian fibration.
\end{proof}

\begin{rem} It is important to note that the interior points of the set 
$\mathcal K_h$ correspond to Lagrangian submanifolds of $M$. On the other hand, 
the points lying on the boundary $\partial \mathcal K_h$ correspond to isotropic 
submanifolds that we can think as degenerate Lagrangian submanifolds. 
We mean that if $x\in \partial\mathcal K$ then
$\omega|_{\mu^{-1}(x)}\equiv 0$. The isotropic submanifolds have 
the form $\mathbb R^r \times \mathbb T^s$ with $0\leq r,s<n$ and 
$r+s<n$
\end{rem}

\begin{cor}
	The hypersurfaces of constant energy of the circular double Sitnikov 
	problem under the momentum map {\rm (\ref{eqn:mom})} correspond to 
	segments of lines with slope $m=-1$ in ${\rm Img}(\mu)\subset
	\mathbb R^2$.
\end{cor}
\begin{proof} Since the circular double Sitnikov problem is a Hamiltonian
system defined on $M\cong \mathbb R^4$, the image of the 
momentum map (\ref{eqn:mom}) is a subset of $\mathbb R^2$ and 
the affine subspaces
perpendicular to $\mathds 1_{\mathbb R^2}=(1,1)$ are the straight lines 
with slope $m=-1$.
\end{proof}

\subsection{Level sets of fixed energy.}

In order to describe the surfaces of constant energy and their foliations,
we consider the separable Hamiltonian function in the following
form
\begin{eqnarray*}
	H= h = h_3+ h_4,
\end{eqnarray*}
where $h_i$ corresponds to an energy level of the circular classical
Sitnikov problem, for each  $i=3,4$. Then we consider the image of the 
momentum map (\ref{eqn:mom}) and finally we construct the 
foliation following the straight line associated to each surface of constant 
energy in ${\rm Img}(\mu)$.

From the solutions for the classical Sitnikov problem \cite{Bel1}, we know
that $h_i$, for $i=3,4$, is defined in $[-2,\infty)$ and the orbits have the 
following behavior: if $-2<h_i<0$ the circular
Sitnikov problem has periodic orbits, for $h_i=0$ it has a parabolic orbit and for
$h_i>0$ it has hyperbolic orbits. Due to the restriction on each
relative energy $h_i$, for $i=3,4$, the total energy $h=h_3+ h_4$ has
the image $(-4,\infty)$
and we have the following topology (see Figure \ref{fig:hfol}):

\begin{figure}[ht]
  \centering
    \includegraphics[scale=.5]{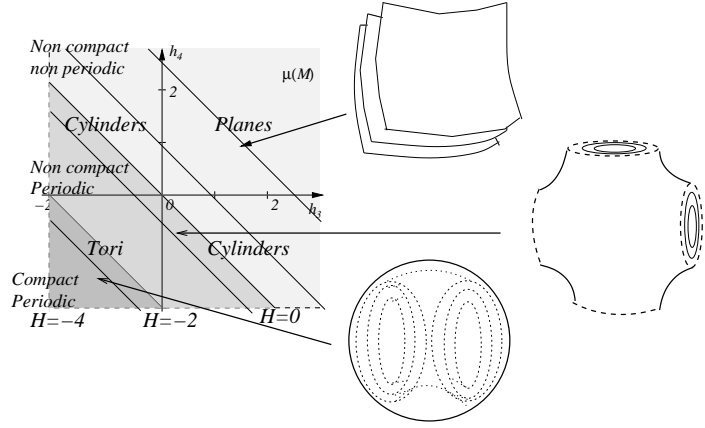}
    \caption{Lines associated to surfaces of constant energy}
    \label{fig:hfol}
\end{figure}

\begin{itemize}
	\item If $h=-4$ this level does not exist in the real problem $q_1=q_2$
   because the secondaries are in the same place at the same time (impossible).
   \item If $-4<h<-2$ the energy levels corresponds to spheres $S^3$ foliated
   by tori $T^2$ and two singular closed curves.
   \item If $h=-2$ the energy surface is a $3$-sphere without four points.
   \item If $-2<h<0$ the energy surfaces are $3$-spheres with four
   discs $D^2$ as boundaries.
   \item If $h=0$ the foliation contains two disjoint cylinders with four planes
   in the middle point (when $h_3=h_4=0$).
   \item If $h>0$ the foliations contains cylinders and planes.
\end{itemize}
A few of these foliations are shown in figure
\ref{fig:hfol}.

The most interesting energy levels are when $h=-2$ and $h=0$
because these are bifurcation values for the topology of the constant
energy surfaces. Other interesting energy levels are $-4<h<-2$ because
the energy surfaces are $3$-spheres foliated by $2$-tori and they have
all the solutions bounded,
and this gives the possibility of finding interesting periodic solutions that
will be preserved under small perturbations of the eccentricity $e$ for the
keplerian solutions of the primary orbits, or perturbations on the mass
parameter $\mu$ of the secondaries.

\section{Periodic solutions for the circular double Sitnikov problem}

At this point we have shown that every solution of the 2+2 Sitnikov problem
has  the form 
\begin{eqnarray}
  \phi(t)= \left( k_3\frac{s(\nu)d(\nu)}{1-2k_3^2s(\nu)^2},2\sqrt{2}k_3c(\nu),k_4\frac{s(\nu)d(\nu)}{1-2k_4^2s(\nu)^2},2\sqrt{2}k_4c(\nu) \right),
\end{eqnarray}
with $\nu=\nu(t)$. When the values of the momentum map are in the third
 quadrant the evolution of the system is bounded.
In this region it is possible to have periodic orbits of the four bodies 
under specific conditions. In what follows we give some definitions and 
we establish the conditions that produce periodic orbits in the circular
double Sitnikov problem.

\begin{defn} We say that $\varphi(t)$ is a periodic solution of period
$\tau$ with $\tau>0$ if $\varphi(t+\tau)=\varphi(t)$ for all $t\in\mathbb R$
and there does not exist $\hat\tau\in (0,\tau)$ such that 
$\varphi(t+\hat\tau)=\varphi(t)$, i.e., $\tau$ is the minimum period.
\end{defn}

Since the solutions of the double Sitnikov problem are in terms of the
Jacobian elliptic functions $sn(\nu,k)$, $cn(\nu,k)$ and $dn(\nu,k)$ which 
are defined on the Riemann surface generated by two primitive periods $2K$ and
$2iK^\prime$  in general, they accept complex arguments and modules. In fact,
these functions are analytic functions in the module $k\in(\mathbb C\setminus \{-1,1\})$,
but just if $k\in(\mathbb R\setminus\{\-1,1\})$ its image is real. Here, $K=K(k)$ is the complete
elliptic integral of first type.
Therefore, the body with position $q_i$ will have a return time $\nu=4K(k_i)$ 
in the rescaled time 
and $t(4K)=T(h_i)$ in the real time $t=t(\nu)$, for $i=3,4$,
\begin{eqnarray}
	T(h)&=& \frac{\sqrt{2}}{2(1-2k^2)}\left[ 2E(k)-K(k)+\Pi(2k^2,k) \right],
	\label{eqn:Tper}
\end{eqnarray}
where $k=\frac{\sqrt{2+h}}{2}$ and $K$, $E$, and $\Pi$ are the complete elliptic 
integrals of first, second and third type respectively.

We will need some more  properties about the function $T(k)=T(h)$, which are
summarized in the following result.
\begin{thm}[\cite{Bel1}]
	Let $T(k)=T(h)$ be the period of the solution of the circular 
	Sitnikov problem with energy $h$; then the following statements hold.
	\begin{enumerate}
		\item $
				\lim_{h\to -2^+} T(h) = \frac{\pi}{\sqrt{2}}.$
		\item $
				\lim_{h\to 0^-} T(h) = \infty.$
		\item $
				\frac{dT}{dh}>0,\hspace{20pt}\forall h\in(-2,0).$
 		\item $
 				\lim_{h\to-2^+}\frac{dT}{dh}=\frac{\pi(1+4\sqrt{2})}{16}.$
 		\item $
 				\lim_{h\to0^-}\frac{dT}{dh}=\infty.$
 \end{enumerate}
 	\label{teo:period}
 \end{thm}
 The proof of this theorem follows directly from the definition of the period 
 $T(k(h))=T(h)$ as function of $h$. We refer the reader to \cite{Bel1}
 for details. 

With these elements, we will characterize the periodic orbits of the 
double Sitnikov problem. We will use the notation $(p,q,n)=1$ to mean
that the {\it greatest common divisor} is ${\rm gcd}(p,q,n)=1$, in other words,
that the three numbers have not common factors at the same time.
\begin{prop}
      For every periodic solution of the double Sitnikov problem there exist 3-plets
$(p,q,n)\in\mathbb Z^3$ such that $(p,q,n)=1$, and $p>\frac{q}{2\sqrt{2}}$ 
and $p>\frac{n}{2\sqrt{2}}$ holds.
  The periods of these solutions are related to the partial energies by 
\begin{eqnarray*}
	\tau = 2p\pi = q T( h_3) = n T( h_4).
\end{eqnarray*}
\end{prop}

\begin{rem} The couples $(p,q)$ and $(p,n)$ are not necessarily coprime,
however, at least one of the three combinations $(p,q),\ (p,n), (q,n)$ must be 
coprime to assure that $(p,q,n)=1$.
\end{rem}

\begin{rem} The energy surface of the double Sitnikov problem that accepts 
      periodic solutions with period 
      $
	      \tau = 2\pi = T(h_3)=T(h_4) 
      $
      is a non compact hypersurface.
    The value of $T(h)$  in $h=-1$ is
\begin{eqnarray*}
	T(-1) &=& \sqrt{2}\left( 2 E(\frac{1}{2})-F(\frac{1}{2})+ \Pi(\frac{1}{2},\frac{1}{2}) \right).
\end{eqnarray*}

The numerical estimation of this value is
\begin{eqnarray*}
	\frac{T(-1)}{2\pi} &=& 0.824429907123718 < 1.
	\label{eq:}
\end{eqnarray*}
Since the function $T(h)$ is an increasing function of $h$ thus
$T(h_i)=2\pi$ is obtained for $h_i>-1$ and $-2<h_3+h_4$,
therefore $\Sigma = H^{-1}(h_3+h_4)$
is not compact (see Figure \ref{fig:hfol}).
\end{rem}

\begin{defn} We say that an energy surface $\Sigma=H^{-1}(h)$
\emph{accepts} a periodic solution if there exists $p,q,n\in \mathbb N$  
with the following properties: 
\begin{enumerate}
	\item[P1.] $(p,q,n)=1$, 
	\item[P2.] $p>\frac{q}{2\sqrt{2}}$, $p>\frac{n}{2\sqrt{2}}$ 
\end{enumerate}
such that 
\begin{eqnarray*}
	\Sigma &=& H^{-1}\left(
	T^{-1}\left(\frac{p}{q}2\pi\right)+
	T^{-1}\left(\frac{p}{n}2\pi\right) \right).
\end{eqnarray*}
We will write $\Sigma_h=H^{-1}(h)$ in order to make clear the dependence 
on $h$.
\end{defn}

We denote the set of fixed energy surfaces that accept periodic
orbits as
\begin{eqnarray*}
	\mathfrak M &=& \left\{ \Sigma = H^{-1}(h_*)| 
 	h_*=T^{-1}\left(\frac{p}{q}2\pi\right) +
	T^{-1}\left(\frac{r}{s}2\pi\right),
	P1,P2\  {\rm holds}
	\right\}.
\end{eqnarray*}

\begin{thm}[\cite{Jim2}]
	In the circular double Sitnikov problem there exists a countable 
	number of energy surfaces $\Sigma\in\mathfrak M$ that 
	contains resonant tori foliated by periodic orbits. Moreover, 
	the set of values $h_*\in H(M)\subset \mathbb R$ 
	such that $\Sigma_{h_*}
	\in \mathfrak M$
is dense in (-4,0) and have zero measure in $\mathbb R$.\label{teo:tori}
\end{thm}

It is a well-known result that resonant tori form a dense set in the 
image of the momentum map. However, Pugh and Robinson proved in 1983 
\cite{Pug1} that generically the periodic orbits of Hamiltonian systems are dense  
in any open set contained in the union of compact and regular energy 
surfaces. Moreover, they argued 
that using a Fubini's argument, this result apply for any given compact 
and regular 
energy surface.
In contrast, last theorem assures that there exists a set of values 
$h_*\in \mathbb R$ of full measure such that $\Sigma_{h_*}\notin\mathfrak M$. 
That is a generic behavior of completely integrable Hamiltonian systems.

The proof of Theorem \ref{teo:tori} is an immediate consequence of the following 
two lemmas that we now state and prove.
\begin{lem}\label{prop:exist}
	For each $n\in \mathbb N$ the circular double Sitnikov problem has  
	periodic solutions of period $2n\pi$.
\end{lem}
 We will just exhibit at least one periodic solution of period
$\tau=2N\pi$. This is immediate from the fact that there exists such periodic
solutions in the circular (classical) Sitnikov problem.

\begin{proof}
For any $N\in\mathbb N$ we can choose the combination
$p=N$ and $q=n=1$ that produce 
\begin{eqnarray*}
	(p,q)=1 &{\rm and}&  (p,n)=1,
\end{eqnarray*}
with
\begin{eqnarray}
	p>\frac{q}{2\sqrt{2}} &{\rm and}&  p>\frac{n}{2\sqrt{2}},\label{eqn:cond}
\end{eqnarray}
and Proposition $2.8$ in \cite{Cor1} assures that there exists
$ h_1,h_2\in (-2,0)$ such that 
\begin{eqnarray*}
	T( h_3)=\frac{2\pi p}{q} &{\rm and}& T(h_4)=\frac{2\pi p}{n}.
\end{eqnarray*}
Then the hypersurface $H^{-1}( h_3+ h_4)$ contains a torus foliated
by a family of periodic orbits with period
\begin{eqnarray*}
	\tau = 2\pi N= T( h_3) = T( h_4).
\end{eqnarray*}
\end{proof}

The following lemma is about the finiteness of resonant tori foliated
by periodic orbits of prescribed period $\tau$. 

\begin{defn} We define the {\it totient} function or {\it Euler's phi 
function} $\varphi(p)$ of an integer $p$  by
\begin{eqnarray*}
	\varphi(p)&=& p \prod_{n|p}\left( 1-\frac{1}{n} \right)
\end{eqnarray*}
where the product runs on all $n$ coprime to $p$. It represents
the number of positive integers less than or equal to $p$ that are 
coprime to $p$.
\end{defn}
\begin{lem}
	For each $N\in\mathbb N$ fixed,
	the circular double Sitnikov problem
	have a finite number of tori foliated
	by periodic orbits with	period $\tau=2N\pi$. The number
\begin{eqnarray*}
	8 N\varphi(N) +  
	\sum_{
	\begin{array}{c}
		q<2\sqrt{2}N,\\
		(N,q)\neq 1
	\end{array}
		}\varphi(q) 
\end{eqnarray*}
is an upper bound (although is not an optimal bound).
\end{lem}
\begin{proof} 
For each $p\in\mathbb N$ fixed there exist $3$-plets $(p,q,n)\in\mathbb N^3$, 
where properties P1 and P2 of Definition 2 holds. Therefore, we search for 
the number $C_p$ of $3$-plets $(p,q,n)=1$ coprimes. It is easy to see that 
for every $q<2\sqrt{2}p$ and $(p,n)=1$, the $3$-plet $(p,q,n)$ does not have
common divisors. These triplets are exactly $(2\sqrt{2}p)\cdot(2\sqrt{2}\varphi(p))=8p\varphi(p)$. 

Additionally, we must add all the couples $(q,n)$ coprime such that $(p,q)$ and $(p,n)$
are not coprime. This means that for each integer $q<2\sqrt{2}p$ with $(p,q)\neq 1$ we must add
the number of coprimes $\varphi(q)$. Then we have 
\begin{eqnarray}
	C_p < 8 p\varphi(p) +  \sum_{
	\begin{array}{c}	
	q<2\sqrt{2}p\\
	(p,q)\neq 1
	\end{array}
	}\varphi(q).\label{eqn:Cp} 
\end{eqnarray}
Finally we must eliminate the elements that are in both sets, however the number 
(\ref{eqn:Cp}) is an upper bound of the 3-plets $(p,q,n)\in\mathbb N^3$ where
properties P1 and P2 hold.

The $3$-plet $(p,q,n)\in\mathbb N^3$ induces a point $x=(2\pi\frac{p}{q},2\pi\frac{p}{n}) 
\in (T(h_3),T(h_4))$ such that the Lagrangian torus $\mathbb T =(\mu^{-1}\circ \mathcal 
T^{-1})(x)$ is foliated by periodic orbits of period $2N\pi$, therefore it is a resonant 
torus $\mathbb T_{Res}\subset M$. 
\end{proof}
\begin{proof}[Proof of Theorem \ref{teo:tori}]
The first part of the theorem is a consequence of the fact that the countable union 
of finite sets is a countable set. Using Lemmas 2 and 3 we have that the number of 
resonant tori are countable, and since each torus belongs to exactly one energy surface,
the set $\mathfrak M$ is countable too.

Now we must to prove that the set of values $h_*$ of energy surfaces with resonant torus 
is dense in $(-4,0)$, and have zero measure there.  We define the map 
$\mathcal T:\mathfrak g^*\to \mathbb R^2$ by
\begin{eqnarray*}
(h_3,h_4)\mapsto \left( \frac{T(h_3)}{2\pi},\frac{T(h_4)}{2\pi} \right).
\end{eqnarray*}

For each rational point 
$y\in {\rm Img} (\mathcal T)$ with $y=(\frac{r}{s},\frac{u}{v})$,
$(r,s)=1$ and $(u,v)=1$, we construct the point $(\frac{ru}{g},\frac{su}{g},\frac{rv}{g})\in
\mathbb N^3$ where $g=gcd(ru,su,rv)$. Since this point fulfills properties P1 and P2 of definition
2, there exists a resonant torus foliated by periodic orbits with period
\begin{eqnarray*}
	\tau = 2\frac{ru}{g}\pi = \frac{su}{g} T(h_3)=\frac{rv}{g}T(h_4).
\end{eqnarray*}
The set of rational values of $\mathcal T$ defined by $RP:={\rm Img}(\mathcal T)\cap \mathbb Q^2$ 
is a dense subset of zero measure in ${\rm Img}(\mathcal T)$.
The mapping $\mathcal T$ is continuous and then $\mathcal T^{-1}(RP)\subset \mathfrak g^*$
is a dense subset in the image of the momentum map $\mu$. Now we construct the function 
$\mathcal H:\mathfrak g^*\to \mathbb R$ such that sends $x=(h_3,h_4)\mapsto h_3+h_4$. It is 
immediate that $\mathcal H(\mathcal T^{-1}(RP))\subset(-4,0)$ is a dense subset by 
continuity, and have zero measure
since $RP$ is a countable set.
\end{proof}

\section{A conjecture}

In this section we use some facts about the transcendental number
theory related to the transcendence of the periods of elliptic functions, 
in  order to
characterize the values $h_{*}\in\mathbb R$ such that we have 
$\Sigma_{h_*}\in\mathfrak M$. 

The results of the last section can be restated as follows

\begin{thm}
  Every point $x\in {\rm Img}(\mu)$ is the projection of a resonant torus foliated
  by periodic orbits of the circular double Sitnikov problem if, and only if
  $\mathcal T(x)$ is a rational point.
\end{thm}
\begin{proof}
Suppose that $\mathfrak T$ is a resonant torus, and $\mu(\mathfrak T)=x\in\mathfrak g^*$
with $x=(h_3,h_4)$.
It means
that there exists a number $\tau\in\mathbb R^+$ such that $\sigma(t+\tau)=\sigma(t)$
holds for every solutions
on $\mathfrak T$. Moreover, $\tau=2n\pi$ for some $n\in\mathbb N$. Since $\tau$ is
the (minimum) period, then there exists $p,q\in\mathbb N$ such that $pT(h_3)=\tau$ and $qT(h_4)=\tau$ 
with $(n,p,q)=1$. We obtain that $\mathcal T(h_3,h_4)=(n/p,n/q)\in\mathbb Q^2$ is a rational point. 
\end{proof}

Now, we want characterize the values of the relative energies $h_3,h_4\in(-2,0)$ which 
produce resonant tori.
Using some relations between the elliptic integrals and functions of Jacobi we have  
the following expression for the complete elliptic integral of third kind 
\begin{eqnarray*}
	\Pi\left(K(k),2k^2,k\right) &=& K(k)E(2k^2,k)-2k^2E(k)
\end{eqnarray*}
(formulae (3.8.32) and (3.6.1) in \cite{Law1}). Therefore, from (\ref{eqn:Tper}) the condition 
for $\mathcal T(x)$ 
be a rational point is equivalent to
\begin{eqnarray}
	\frac{T_i}{2\pi}=\frac{1}{4\pi\sqrt{2}(-h_i)}\left[ 2(k_i^\prime)^2E(k_i)-(1-E(2k_i^2,k_i))
	K(k_i) \right]\in \mathbb Q, \label{eqn:final}
\end{eqnarray}
where $E(k_i), i=3,4$ is the complete elliptic function of second kind, 
$(k_i^\prime)^2=1-k_i^2$ is the complementary modulus, and
$E(2k_i^2,k_i)$ is the incomplete elliptic function of second type 
with argument $2k_i^2$ and modulus $k_i$. In the last formula, the ratio $T_i/(2\pi)$ is expressed in terms of
elliptic functions of first and second kind only. Then, it is possible to apply
some results on transcendental number theory due to Schneider \cite{Sch1} in order to characterize the 
values of $h_i$ and $k_i$ such that expresion (\ref{eqn:final}) holds.

\begin{conj} If the circular double Sitnikov problem has a periodic 
solution with period $\tau=2n\pi, n\in\mathbb N$ then the relative
energy values belongs to the field $\mathbb Q\langle k^*\rangle$
where $k^*\in\{\mathbb R \setminus \mathbb A\}$. This field is an 
extension with degree of transcendence 1 over $\mathbb Q$.
\end{conj}
It means that all the constant energy values where the resonant tori lie 
are algebraically dependent.







\section*{Acknowledgment}
Research partially done during an academic stay of the first
author at the IMCCE institute of the {\it Observatoire
de Paris} and supported by CoNaCyT through Ph.D. fellowship No. 184728.



\begin{thebibliography}{1}

   \bibitem{Ale1}Alekseev, \emph{Quasirandom Dynamical Systems I, II, III}, Math USSR Sbornik, 5, pp 73-128; 6, pp 505-560; 7, 1-43, 1968.
   \bibitem{Bel1}E. Belbruno, J. Llibre, M. Oll\'e, \emph{On the Families of Periodic Orbits which Bifurcate from the Circular Sitnikov Motions}, Celestial Mechanics, No 96, 1994, pp 99-129.
   \bibitem{Cha1} J. Chazy, \emph{Sur l'allure du mouvement dans le probl\`eme des trois corps quand le temps cro\^{\i}t ind\'efiniment.} Ann. Sci. de l'\'E.N.S., S\'er. 3, 39, Paris, 1922, pp 29-130.
   \bibitem{Cor1} M. Corbera, J. Llibre, \emph{On Symmetric Periodic Orbits of the Elliptic Sitnikov Problem Via the Analytic Continuation Method}, Con. Math., 292, A.M.S., 2002, pp 91-127.
   \bibitem{Dan1}H. Dankowicz, P. Holmes, \emph{The Existence of Transverse Homoclinic Points in the Sitnikov Problem}, Journal of Differential Equations, Volume 116, 1995, pp 468-483.
   \bibitem{Dvo1}R. Dvorak, Yi Sui Sun, \emph{The phase space structure of the extended Sitnikov problem}, Celestial Mechanics and Dynamical Astronomy, 67, 1997, pp 87-106.
   \bibitem{Gar1}A. Garc\'{\i}a, E. P\'erez-Chavela, \emph{Heteroclinic Phenomena in the Sitnikov Problem},  Ham. Sys. and Cel. Mech. (HAMSYS-98), World Scientific, 2000, pp 174-185.
   \bibitem{Gui1}V. Guillemin, S. Stengberg \emph{Convexity Properties of the Moment Mapping}, Invent. Math No. 69, Springer-Verlag, 1982 pp 491-513.
    \bibitem{Han1}H. Hancock, \emph{Lectures on the Theory of Elliptic Functions}, John Wiley \& Sons, New York, 1910.
    \bibitem{Hof1}H. Hofer, E. Zender, \emph{Symplectic Invariants and Hamiltonian Dynamics}, Birkh\"auser, New York, 1994.
   \bibitem{Jim2}H. Jim\'enez-P\'erez, E. Lacomba \emph{On the periodic orbits of the double Sitnikov problem}, C. R. Acad. Sci. Paris, Ser. I, 347, 2009, pp 333-336.
   \bibitem{Lac1}E. Lacomba, J. Libre, E. P\'erez-Chavela, \emph{The Generalized Sitnikov Problem}, Contemporary Mathematics, Volume 292, 2002, pp 147-158
   \bibitem{Law1}D. F. Lawden,  \emph{Elliptic Functions and Applications}, Applied Mathematical Sciences, 98, Springer-Verlag, 1989.
   \bibitem{Mac1}W.D. MacMillan, \emph{An Integrable Case in the Restricted Problem of Three Bodies}, Astronomical Journal, Volumen 27, 1911, pp 11-13.
   \bibitem{Mos1} J. Moser, \emph{Stable and Random Motions in Dynamical Systems}, Annals of Math Studies 77, Princeton Univ. Press, New Jersey, 1973.
   \bibitem{Pug1}C. Pugh, C. Robinson, \emph{The $C^1$ closing lemma, including Hamiltonians}, Ergod. Th. Dynam. Sys. 3, 1983, pp 261-313.
   \bibitem{Sch1}T. Schneider \emph{Einf\"uhrung in die Transzendenten Zahlen}, Berlin, Springer-Verlag, 1957.
   \bibitem{Sit1}K. Sitnikov, \emph{Existence of oscillating motions for the three-body problem}, Dokl. Akad. Nauk. Volume 133, No. 2, URSS 1960, pp 303-306.	
   \bibitem{Sou1}P. S. Soulis, K. E. Papadakis, T. Bountis \emph{Periodic orbits and bifurcations in the Sitnikov four-body problem} Cel. Mech. and Dyn. Astr. 100, 2008, pp 251-266.
\end{thebibliography}
\end{document}